\documentclass[top=30truemm,bottom=30truemm,left=25truemm,right=25truemm]{article}

\usepackage[english]{babel}
\usepackage[utf8x]{inputenc}
\usepackage[T1]{fontenc}

\usepackage{algorithm}
\usepackage{algorithmic}
\usepackage{bm}
\usepackage{amsfonts}
\usepackage{amsthm}
\usepackage{amsmath}
\usepackage{amssymb} 

\newtheorem{assumption}{Assumption}
\newtheorem{definition}{Definition}
\newtheorem{theorem}{Theorem}
\newtheorem{lemma}{Lemma}
\newtheorem{fact}{Fact}

\numberwithin{equation}{section}

\newcommand{\Prob}{\mathbb{P}}
\newcommand{\Expect}{\mathbb{E}}
\newcommand{\Real}{\mathbb{R}}
\newcommand{\Natural}{\mathbb{N}}
\newcommand{\Normal}{\mathcal{N}}
\newcommand{\Unif}{\mathrm{Unif}}
\newcommand{\Beta}{\mathrm{Beta}}

\newcommand{\Dp}{\mathcal{D}_P}
\newcommand{\Di}{\mathcal{D}_I}
\newcommand{\Da}{\mathcal{D}_A}
\newcommand{\Np}{N_P}
\newcommand{\Ni}{N_I}
\newcommand{\Na}{N_A}
\newcommand{\val}{v}

\newcommand{\pp}{p_P}
\newcommand{\hmu}{\hat{\mu}}
\newcommand{\hsigma}{\hat{\sigma}}
\newcommand{\mugap}{\mu_{\mathrm{gap}}}

\newcommand{\nn}{\nonumber\\}
\newcommand{\bPhi}{\bar{\Phi}}
\newcommand{\bPhiInv}{\bar{\Phi}^{-1}}
\newcommand{\e}{\mathrm{e}}

\newcommand{\dKL}{d_{\mathrm{KL}}}

\newcommand{\conserv}{{\mathrm{con}}}
\newcommand{\hmutopk}{\hat{\mu}_{\mathrm{top}}}
\newcommand{\bT}{\bar{T}}
\newcommand{\erf}{\mathrm{erf}}

\usepackage{amsmath}
\usepackage{graphicx}
\usepackage[colorinlistoftodos]{todonotes}
\usepackage[colorlinks=true, allcolors=blue]{hyperref}

\usepackage{subfigure}
\newlength{\subfigwidth}
\newlength{\subfigcolsep}
\usepackage{color}
\usepackage{ifthen}
\newcommand{\COMM}[2]{{
\ifthenelse{\equal{#1}{JK}}{\color{blue}}{
\ifthenelse{\equal{#1}{TM}}{\color{red}}{
\ifthenelse{\equal{#1}{AA}}{\color{magenta}}{
\ifthenelse{\equal{#1}{BB}}{\color{cyan}}}}}
{}[#1: #2]
}}

\usepackage{url}
\usepackage{hyperref}

\title{A Simple Way to Deal with Cherry-picking}
\author{Junpei Komiyama \\ \href{mailto:junpei@komiyama.info}{junpei@komiyama.info} 
   \and Takanori Maehara \\  \href{mailto:takanori.maehara@riken.jp}{takanori.maehara@riken.jp} }

\date{October 2018}

\begin{document}

\maketitle

\begin{abstract}
Statistical hypothesis testing serves as statistical evidence for scientific innovation.
However, if the reported results are intentionally biased, hypothesis testing no longer controls the rate of false discovery.
In particular, we study such selection bias in machine learning models where the reporter is motivated to promote an algorithmic innovation.
When the number of possible configurations (e.g., datasets) is large, we show that the reporter can falsely report an innovation even if there is no improvement at all. 
We propose a ``post-reporting'' solution to this issue where the bias of the reported results is verified by another set of results. 
The theoretical findings are supported by experimental results with synthetic and real-world datasets.
\end{abstract}

\section{Introduction}

Concern about the reproducibility of scientific results is mounting.
The results of a survey \cite{baker2016} of over 1,576 researchers who took a brief online questionnaire on reproducibility in research confirmed that results are rarely reproducible in many areas of science.
Although failure of reproducibility does not necessarily mean the results are false, 80\% of the researchers consider there is a slight (38\%) or significant (52\%) crisis regarding reproducibility.
Analyses of published phychology~\cite{aac4716} research and cancer biology research\cite{begley2012drug} found that only 40\% and 10\% of the reported results were reproducible, respectively. 

Reproducibility of scientific results has attracted much attention in the field of computer science. 
In fact, workshops focused on reproducibility are frequently held in many branches of computer science, such as machine learning\footnote{https://mltrain.cc/events/enabling-reproducibility-in-machine-learning-mltrainrml-icml-2018/} and computer networks\footnote{http://conferences.sigcomm.org/sigcomm/2017/workshop-reproducibility.html}. 
Unlike natural science experiments, anyone can replicate the results of studies in computer science if the authors of the study provide a way to run the same program on the same computing environment.
Accordingly, significant attention \cite{Sonnenburg2007,citeulike:8011561} has been paid to software engineering solutions that enable replication and urge authors to publish their source codes and datasets along with their papers. Nevertheless, perfect replicability of published results only provides a limited guarantee of reproducibility \cite{replicarepro,Cohen2018ThreeDO}.

Guaranteeing the reproducibility of a hypothesis requires not only replicating the same experiment under identical settings but also conducting experiments with different settings to examine the generality of the hypothesis.
In this paper, we consider reproducibility in the context of machine learning where we seek algorithms that conduct statistical inference from datasets.
When an algorithmic innovation is reported, we expect that it improves not only the results on the published datasets but also those in many other datasets (and many different settings such as different hyperparameters). 
Usually, the amount of improvement varies from one setting to the next.
Because the authors of the innovative algorithm are required to show a significant amount of improvement, there exists a motivation for \emph{cherry-picking} the datasets that fit well with the proposed algorithm.  
Such selection bias can sometimes devastate the ground under the hypotheses: \cite{DBLP:conf/emnlp/ReimersG17} showed a case where a proclaimed improvement on natural language models is merely due to the choice of initial seeds in deep neural models. 

Statistical testing is widely used as evidence for a scientific hypothesis.
For the reported results, a statistical test associates a scalar value called the $p$-value that indicates how unlikely the experimental results are under the null hypothesis (i.e., no improvement).
A smaller $p$-value implies the null hypothesis is not the case, and thus the original hypothesis is likely to be true. 
One usually defines a predefined $p$-value threshold $\alpha$ (typically $\alpha = 0.05$ or $0.005$), and a hypothesis of its $p$-value $\alpha$ or smaller is considered to be a ``significant'' ones.
Ideally, a statistical test of $\alpha$ can be used to keep the ratio of false findings at $\alpha$ or smaller. However, when data is biasedly selected, the $p$-value no longer controls the ratio of false findings. We show to what extent the author can falsely claim an innovation when there is no actual improvement and consider a way to avoid such a false claim.

\paragraph{Contributions:} The contributions of this paper lie in the following aspects.
\begin{itemize}
    \item 
    \textbf{Does biased selection harms statistical guarantees?} When the number of reported datasets is small, the reporting procedure may be highly biased toward the cherry-picked best results.
    An obvious way to prevent such biased selection is to require the reporter to conduct experiments on more than one dataset.
    However, rather paradoxically, we show that false reporting becomes even easier when the number of available datasets is large (Section \ref{sec_cherrypick}).
    
    \item 
    \textbf{Can we consider a $p$-value under selection bias?} To address the issue of false reporting, we consider a $p$-value that takes selection bias into consideration (Section \ref{sec_conserv}).
    Although such a conservative $p$-value prevents us from claiming a false improvement, we argue that there are two problems with this version of the $p$-value.
    Firstly, one needs to know an estimate of the number of datasets from which the reporter selects.
    Secondly, such a procedure is extremely conservative: in return for statistical correctness, it loses a statistical power to claim an innovation when the improvement is true but not very large.
    
    \item 
    \textbf{How can we deal with the issue of selection bias without compromising statistical power?} A conservative $p$-value (Section \ref{sec_conserv}) sacrifices statistical power, and thus a reporter may not be able to claim a non-negligible portion of true improvements.
    To address this issue, we consider a post-reporting model in Section \ref{sec_inspector}.
    In this model, the reporter publishes a standard $p$-value.
    As demonstrated in Section \ref{sec_cherrypick}, the standard $p$-value is vulnerable to selective reporting.
    However, we consider ``an inspector'' who can detect the bias of published datasets by double-checking with additional datasets.
    We show that, when the number of the published datasets is sufficiently large, the inspector can detect a false innovation. 
    
    \item 
    \textbf{Finite-time analysis on the order statistics:} We quantify selective bias in terms of the order statistics of the normal distribution.
    Finite-time analysis on the order statistics is of independent interest.
    In particular, unlike existing analyses of normal order statistics \cite{boucheron2012}, our analysis uses the inverse survival function $\bPhiInv$ to explicitly represent the normal order statistics. 
    Here, the standard McDiarmid's inequality cannot be used in our analysis because the sensitivity of $\bPhiInv(\alpha)$ diverges as $\alpha \to 0$. 
    Note that extreme value theory \cite{evtbook} does not be directly applied to our case when the number of selected datasets is large (i.e., $\Np/\Na$ does not go to zero).
    \item \textbf{Empirical evaluation of selection bias:} In Section \ref{sec_experiment}, we describe simulations with synthetic and real-world datasets. In particular, the latter datasets involve classification tasks and compare logistic regression (LR) and gradient boosting tree (GBT) classifiers (on $66$ datasets). Despite that GBT outperforms LR on most of the datasets, we show that although cherry-picking the datasets enables us to make a false claim that LR outperforms GBT, such biased selection is detectable by employing the inspector.
\end{itemize}

\subsection{Related work}
\label{sec_relwork}

\textbf{Sample selection bias} has been studied in a branch of the statistical machine learning \cite{Huang:2006:CSS:2976456.2976532,DBLP:conf/alt/CortesMRR08} that deals with the problem of a biased sample distribution in a single dataset. While these lines of work consider bias correction on a single dataset, we consider selection bias among many datasets.

\textbf{Multiple testing:} When the number of hypotheses is large, the standard procedure of hypothesis testing yields a non-negligible amount of false findings, and thus, a multiple testing correction is required for controlling the number of false findings \cite{benjaminimtest,DBLP:conf/nips/YangRJW17}. A classic paper by \cite{ioannidis} considered a Bayesian model that explains why false discoveries frequently occur when the number of underlying experiments is large. Unlike such multiple testing procedures that consider various hypotheses simultaneously, we consider a single hypothesis (i.e., whether an algorithm is innovative or not) by using multiple sets of evidence (datasets). 

\textbf{Selective inferences:} As discussed in \cite{fithian2015optimal}, selective inference can consider selective bias by developing a conservative confidence interval and associated $p$-value (Figure 3 therein). There are two seminal differences between selective inference and our model: First, unlike existing models such as \cite{DBLP:conf/aistats/NieTTZ18}, our model does not require an explicit form of bias underlying the publishing process that is virtually impossible to replicate. Second, we consider a post-publication process where the inspector checks the bias of the submitted results and thus does not compromise its statistical power.

\section{Problem Setup}
\label{sec_setup}

We consider standard machine learning tasks such as classifications and regressions. Let $\Da = \{1,2,\dots,\Na\}$ be the indices of each dataset. The performance of an algorithm on each dataset $i \in \Da$ is measured by using standard evaluation procedure such as cross-validation. Our framework involves a reporter who bring a new algorithm, and our main concern is whether the new algorithm improves the measured performance or not on these datasets. As the number of all the datasets $\Na$ is large, the reporter brings a subset of the datasets (subindices) $\Dp \subseteq \Da$, and reports the measured improvement on the performance $\val_i \in \Real$ associated with the $i$-th dataset for each $i \in \Dp$. For example, when the task is the classification, $\val_i$ corresponds to the gap of the prediction accuracy between the new and the existing algorithms on the $i$-th dataset.
This paper assumes that the reporter makes biased selection of $\Dp$ to overpromote the amount of the improvement brought by the new algorithm. Our main concern is whether or not such biased selection leads to a false claim of improvement.

In the following, we state statistical assumptions.
\begin{assumption}{\rm (Normality)}
For each $i \in \Da$, we assume
\begin{equation}
\val_i \sim \Normal(\mu, \sigma^2),
\label{ineq_norm_drawn}
\end{equation}
where $\mu \in \Real$ indicates the magnitude of the true improvement brought by the new algorithm. 
\end{assumption}
Similar to standard hypothesis testing, we assume each result is normally distributed. For ease of discussion, we assume that $\sigma$ is known: In this case, without loss of generality we can assume $\sigma = 1$ by a proper scaling\footnote{The case of $\Normal(\mu, \sigma^2)$ is equivalent to the case of $\Normal(\mu/\sigma, 1)$.}. We later discuss the case of unknown variance $\sigma^2$ (Section \ref{sec_unkvariance}).

\subsection{Testing}

This section describes the statistical testing methods for $\Dp$.
Arguably, one of the most important assumptions in using statistical testing is that the samples used to construct statistics are drawn uniformly from distributions. 
In our setting, the statistical test assumes $v_i \sim \Normal(\mu, \sigma)$ for each $i \in \Dp$.
Under this assumption, the null hypothesis is
\[
 H_0: \mu=0.
\]
Rejecting the null hypothesis implies that the following alternative hypothesis is supported.
\[
 H_1: \mu>0.
\]
The null hypothesis $H_0$ states that there is no innovation, whereas the alternative hypothesis $H_1$ states that the innovation is true.

Notice that the the mean of $n$ independent standard normal variables follows $\Normal(\mu, 1/\sqrt{\Np})$.
Whether or not the null hypothesis $H_0$ is rejected at a given significance level $\alpha$ is determined by the $p$-value:
\begin{definition}{\rm (standard $p$-value, one-sided normal test)}
Let $\hmu_P = (1/\Np)\sum_{i \in \Dp} \val_i$ be the empirical innovation in the published datasets.
Under the assumption $\val_i \sim \Normal(\mu, \sigma)$ (i.e., unbiased selection), 
\[
  \hmu_P \sim \Normal(\mu, 1/\Np).
\]
and, the associated $p$-value is
\[
  \pp(\hmu_P) = \bPhi(\hmu_P  \sqrt{\Np} )
\]
where $\phi(x) = (1/\sqrt{2 \pi}) \exp(-x^2/2)$ is the standard normal density function, and $\bPhi(x) = \int_{x}^{+\infty} \phi(x) dx$ is the survival function of the standard normal distribution.
We consider the new algorithm made a significant improvement at level $\alpha \in (0,1)$  (or just ``significant'') if $\pp(\hmu_P) \le \alpha$.
\end{definition}

As we discuss in the next Section \ref{sec_cherrypick}, these assumptions may be violated by biased selection of $\Dp$ by the reporter, and in such a case the $p$-value no longer controls the confidence level. 

\section{Standard $p$-value under Selection Bias}
\label{sec_cherrypick}

This section considers the case that $\Dp$ involves some selection bias.
Formally, we consider a biased reporter that knows $\Da$ and associated statistics $\{\val_i\}_{i=1}^{\Na}$. The reporter selects $\Dp$ so as to yields a small $p$-value. 

To obtain some idea on the power of biased selection, we first consider the case of $\Np = 1$.
\begin{theorem}{\rm (Biased reporter, the case of $\Np=1$)}
Let $\Np = |\Dp| = 1$. Assume that $\mu =0$.
Then, with probability $1 - (1-\alpha)^{\Na}$, the reporter can choose $\Dp$ such that $\pp(\Dp) \le \alpha$. 
\end{theorem}
The theorem implies that the $p$-value no longer controls the level of false discovery when a reporter can select an index among many since $1 - (1-\alpha)^{\Na}$ is much larger than $\alpha$ when $\alpha \ll 1$:
If $\Na = O( \log(1/(1-\alpha)) )$, then the biased reporter can falsely claim the innovation brought by the new algorithm even if there is no true improvement (i.e., $\mu=0$). The proof is straightforward since the probability of each $\val_i$ exceeding $\alpha$-quantile is $\alpha$, and thus the probability that at least one $i \in \Da$ such that $\val_i \le \alpha$ exists is $1-(1-\alpha)^{\Na}$.   

We next consider the case of general $\Np$. The following theorem states that, even though the reporter must choose $\Dp$ among $\Da$, he still has a strong power enough to falsely claim the innovativeness of the new algorithm.  
\begin{theorem}{\rm (Biased reporter, the case of general $\Np$)}
Let $\mu = 0$. 
For any $\alpha \in (0,1)$, $\delta \in (0,1)$, $\epsilon \in (0,1/2)$ 
and $\Np, \Na \in \Natural$, if $\Np/\Na \le 1/2 - \epsilon$, $\Np \ge \max(\log(1/\delta)/(2\epsilon^2), 8\log(1/\alpha)/\epsilon^2)$, then with probability at least $\delta$, the reporter can choose $\Dp$ such that $\pp(\Dp) \le \alpha$.
\label{thm_cherrypick}
\end{theorem}

\begin{figure}[t!]
\begin{center}
  \setlength{\subfigwidth}{.99\linewidth}
  \addtolength{\subfigwidth}{-.99\subfigcolsep}

  \begin{minipage}[t]{\subfigwidth}
  \centering
  \subfigure[Magnitude of $\hmu_P$ that is required to be significant.]{
    \includegraphics[scale=0.15]{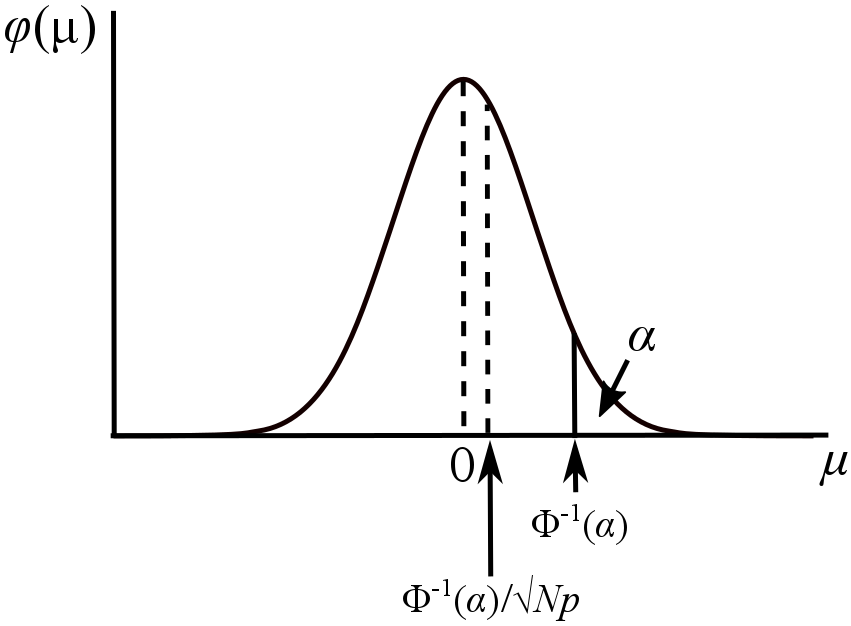}
  }
  \end{minipage}\hfill
  \begin{minipage}[t]{\subfigwidth}
  \centering
  \subfigure[Bias of $\hmu_P$ that the top-$\Np$ yields.]{
    \includegraphics[scale=0.15]{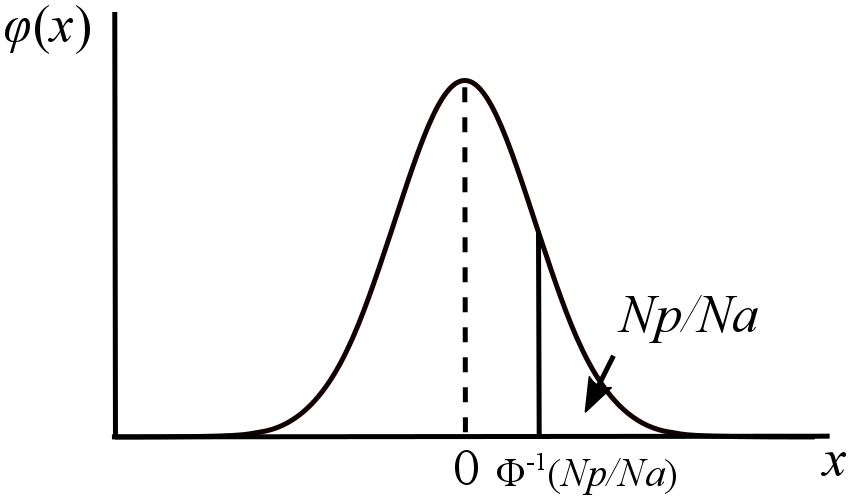}
  }
  \end{minipage}\hfill

\end{center}
\label{fig_quantile}
\caption{\textbf{(a)} On one hand, the standard $p$-value only requires $\hmu_P \ge \bPhiInv(\alpha)/\sqrt{\Np}$ for a test to be significant. Note that this threshold can be arbitrary small for large $\Np$.
\textbf{(b)} On the other hand, $\hmu_P$ of the cherry-picked top-$\Np$ is larger than $\bPhiInv(\Np/\Na) > 0$ on average. 
}
\end{figure}%
Due to space limitations, the proofs of Theorem \ref{thm_cherrypick} and subsequent theorems are shown in appendix.
The intuition behind Theorem \ref{thm_cherrypick} is represented in Figure \ref{fig_quantile}. The theorem states that, for sufficiently large $\Np$ it holds that $\bPhiInv(\alpha)/\sqrt{\Np} < \bPhiInv(\Np/\Na)$ and thus one can expect it is considered to be a significant improvement for any $\Np, \Na$ such that $\Np/\Na < 1/2$. 
The reporter can be very powerful even though we require large size of $\Np$, and increasing the number of datasets does not solve the problem of selective reporting at all. On the contrary, a large $\Np$ even makes the problem worse. Even when the reporter submits the result with a half of all datasets (i.e., $\Np \approx (1/2)\Na$), the selective reporting enables a false claim of improvement.

\section{Conservative $p$-value}
\label{sec_conserv}

Section \ref{sec_cherrypick} revealed that the standard $p$-value no longer controls the rate of false claim under biased selection of $\Dp$.
Since $\Dp$ that maximizes $\hmu_P$ is the one that selects the top-$\Np$ elements among $\Da$, a ``conservative'' version of $p$-value can be obtained by calculating the $\alpha$-quantile of the top-$\Na$ elements. Let $\hmutopk$ be a random variable such that
\begin{align}
x_i &\sim \Normal(0, 1), \\ 
\hmutopk &= \frac{1}{\Np} \sum_{i=1}^{\Np} x_{(i)}
\end{align}
where $x_{(i)}$ is the $i$-th largest among $\{x_1 , \dots , x_{\Na}\}$.
The conservative $p$-value $p_P^\conserv$ is defined as:
\begin{equation}
p_P^\conserv(\hmu) = \Prob[ \hmu \ge \hmutopk ]
\end{equation}
where the expectation is taken with respect to the random variables $\{x_1 , \dots , x_{\Na}\}$. 
One may consider $p_P^\conserv$ as the survival function of $\hmutopk$.
Although we presume that  $p_P^\conserv$ cannot be represented by a closed formula, a Monte Carlo method yields a reasonable estimator of $p_P^\conserv$ for a moderate value of $p_P^\conserv$. 

There are two concerns about the practical use of $p_P^\conserv$.
Firstly, calculating $p_P^\conserv$ requires a reasonable estimate of $\Na$: For example, when we consider classification, we need to find how many datasets from which the reporter picks the results, which is generally hard to infer. Secondly, $p_P^\conserv$ is massively conservative and compromising the statistical power of finding a true hypothesis. The following theorem gives us a measure of how much statistical power we lose when we use $p_P^\conserv$ instead of the standard $p$-value $p_P$.
\begin{theorem}{\rm (Statistical power of $p_P^\conserv$)}
Assume that $\Dp$ consists of samples of size $\Np$ that are i.i.d from $\Normal(\mu, 1)$.
Let $r = \Np /(\Na+1)$.
For any $\alpha \in (0,1/2)$, $\delta \in (0,1)$, 
if $\Np, \Na \ge 2$ satisfy $\Np < 2 \Na$ and $\mu \le \bPhiInv(r) - \bPhiInv(\delta)$,
 then with probability $1 - \delta$, $p_P^\conserv(\hmu_P) > \alpha$.
\label{thm_losepower}
\end{theorem}%
Theorem \ref{thm_losepower} states that using the conservative $p$-value disables the reporter to make a   claim of innovation when the improvement is not very large $\mu$, which implies that the statistical power is compromised: As $p_P^\conserv$ considers a top-$\Np$ quantile among null distribution $\Normal(0, 1)$, it  concentrates around $\Np/\Na$. When the magnitude of true improvement $\mu>0$ is smaller than $\approx \bPhiInv(r) \approx \bPhiInv(\Np/\Na)$, it cannot claim the improvement, which is the cost of using the conservative $p$-value $p_P^\conserv$.
We empirically discuss the statistical power of $p_P^\conserv$ in the experimental section (Section \ref{sec_experiment}).

\section{Inspector Model}
\label{sec_inspector}

Section \ref{sec_cherrypick} shows that the standard hypothesis testing is not robust to selective reporting.
Although the conservative $p$-value proposed in Section \ref{sec_conserv} controls the level of false reporting, there are some concerns in the practical use of the conservative $p$-value.
This section gives an alternative solution to the problem of the false reporting.
Based on the alternative $p$-value, Section \ref{sec_inspector_ss} introduces an inspector who detects the biased selection of the reporter. 
Under mild assumptions, we show that the inspector can detect any biased selection.

\subsection{Requiring minimum innovation}
\label{sec_mininnov}

Remember that the null hypothesis $H_0$ introduced in Section \ref{sec_setup} assumes that the performance improvement brought by the new algorithm is zero. The alternative hypothesis is $H_1: \mu>0$, which implies the improvement is positive but can be arbitrarily small.
In this section, instead of $H_0$, we require the minimum improvement $\mugap > 0$ to the reporter. That is, for given $\mugap>0$, the null hypothesis is
\[
H_0^{\mugap}: \mu = \mugap 
\]
against the one-sided alternative hypothesis
\[
H_1^{\mugap}: \mu > \mugap.
\]
Whether or not the null hypothesis $H_0^{\mugap}$ is rejected or not is determined by the following $p$-value:
\[
 \pp^{\mugap}(\hmu_P) = \bPhi( (\hmu_P - \mugap ) \sqrt{\Np} ).
\]
While $H_1$ requires the  positive innovation $\mu>0$ that can be arbitrarily small, $H_1^{\mugap}$ requires $\mu$ is larger than $\mugap$.

The following theorem implies that, under sufficiently large $\mugap$, the reporter cannot make false innovation:
\begin{theorem}
Let $\mu=0$. For any $\alpha, \delta \in (0,1)$, $\Np < \Na \in \Natural$.
If 
\[
\mugap \ge 3 \sqrt{2 \log\left(\frac{1}{\delta}\right)} + 7 \sqrt{2 \log\left(  \frac{e \Na}{\Np}\right)} + \frac{\pi}{2 \Np} ,
\]
 then with probability at least $1 - \delta$, for any choice of $\Dp$, $\pp^{\mugap}(\hmu_P) > \alpha$ holds.
\label{thm_mininnov}
\end{theorem}
Theorem \ref{thm_mininnov} states that, for sufficiently large $\mugap = O(\sqrt{\log(\Np/\Na)}) = O(\bPhiInv(\Np/\Na))$ the reporter cannot claim significant improvement when there is no improvement (i.e., $\mu = 0$). 
Although the proof is technically involved, the Theorem is intuitively understood as follows: To maximize $\hmu_P$, the reporter tries to pick the top-$\Np$ among $\{\val_i\}$, and in this case $\hmu_P$ asymptotically converges to $\int_{x=\bPhiInv(\Np/\Na)}^\infty \psi(x) x dx$ (Figure \ref{fig_quantile}, right) for large $\Na$. When this value is no bigger than $\mugap$, the reporter cannot claim an innovation.

\subsection{Inspector and equal mean testing}
\label{sec_inspector_ss}

Section \ref{sec_mininnov} introduced an alternative testing that requires minimum innovation $\mugap > 0$. 
However, the minimum innovation does only provide a limited solution to the problem of the selective reporting because it is generally hard to determine a proper value of $\mugap$:
On the one hand, if $\mugap$ is excessively large, it ignores an innovation smaller than $\mugap$, and loses its statistical power. On the other hand, if $\mugap$ is not sufficiently large, the reporter can still make a false claim. Theorem \ref{thm_mininnov} requires $\mugap = O(\log(\Na/\Np))$ to avoid a false claim, which diverges as $\Na/\Np \rightarrow \infty$. This is not very convenient guarantee when the number of possible datasets is large.  
To solve this problem, this section introduces an inspector who, after the report of $\Dp$, checks whether or not $\Dp$ is biased by drawing unbiased samples from $\{\val_i\}$. 

Formally, let $\Di = \{1,2,\dots,\Ni\}$ be a set of indices with its size $|\Di| = \Ni$. We assume the elements of $\Di$ are i.i.d. samples from the same distribution as $\Da$. That is, the associated improvement for each $i$ of $\Di$ is
\begin{equation}
 \val_i^{I} \sim \Normal(\mu, 1),
\end{equation}
where $\mu$ is the same mean improvement\footnote{We define $\{\val_i^{I}\}$ in $\Di$ and $\{\val_i\}$ in $\Da$ are different samples from $\Normal(\mu, 1)$ just for the ease of analysis in Theorem \ref{thm_inspector}.} as \eqref{ineq_norm_drawn}.

Let $\hmu_I = (1/\Ni)\sum_{i \in \Di} \val_i^{I}$ be the mean improvement observed in $\Di$. 
The standard test that compares the two means utilizes the two-sample z-statistic 
\[
 Z = \frac{\hmu_P - \hmu_I}{ \sqrt{ \frac{1}{\Np} + \frac{1}{\Ni} } }
\]
which, assuming that the two distributions share a common mean, follows the standard normal distribution. 
The corresponding null and alternative hypotheses are the following.
\begin{align*}
 H_0^Z: & \hmu_P = \hmu_I \nn
 H_1^Z: & \hmu_P > \hmu_I,
\end{align*}
and whether or not the null hypothesis $H_0^Z$ is rejected is determined by the $p$-value
\[
 p_Z = \bPhi\left( \frac{\hmu_P - \hmu_I}{ \sqrt{ \frac{1}{\Np} + \frac{1}{\Ni} } } \right).
\]
The smaller the value of $p_Z$ is, the more unlikely that the two samples share a common mean. 
The inspector claims that $\Dp$ is biased with confidence $\beta \in (0,1)$ if $p_Z \le \beta$. 
The following theorem uncovers the power of the inspector.
\begin{theorem}
Let $\mu=0$ and $\Ni = \Np$. 
Assume that the submitted $\Dp$ is such that $\pp^{\mugap}(\hmu_P) \le \alpha$ (i.e., innovative at significance level $\alpha \in (0,1)$ with $\mugap > 0$).
For any $\alpha, \delta, \beta \in (0,1)$, if $\mugap \ge (1/\sqrt{\Np}) (\bPhi^{-1}(\beta) + \sqrt{1/2} \bPhi^{-1}(\alpha) + \sqrt{1/2}  \bPhi^{-1}(\delta))$, then with probability $1 - \delta$, $p_Z \le \beta$.
\label{thm_inspector}
\end{theorem}
Unlike Theorem \ref{thm_mininnov} that requires $\mugap = O(\log(\Na/\Np))$, Theorem \ref{thm_inspector} states that the inspector is able to detect any biased selection if $\mugap = O(1/\sqrt{\Np})$ that can be arbitrarily small when $\Np$ is large.  
One may interpret this result as a powerfulness of a post-reporting validation process:
While Section \ref{sec_conserv} shows a hardness of preventing a false reporting in some sense, the inspector model detects a selection bias when $\Np$ is moderate. 
Besides, it is much easier to adapt because its modification to the reporter's hypothesis testing procedure (Section \ref{sec_mininnov}) is modest compared with the conservative $p$-value (Section \ref{sec_conserv}).

\section{Hypothesis Testing with Relaxed Assumptions}

Up to now, we have assumed that $\val_i$ is normally distributed and the variance $\sigma$ is known. In practice, it is sometimes plausible to relax these assumptions. 

\subsection{Parametric testing with unknown variance}
\label{sec_unkvariance}

In this section, we show that a similar testing procedure to the known variance case can be conducted when we do not know the variance $\sigma^2$.  
Let $\hsigma_P^2 = (1/(\Np-1))\sum_{i \in \Dp} \val_i^2$ be a unbiased estimated variance. A standard test on the mean of a normal distribution with unknown variance utilizes the fact that the statistics 
\[
\sqrt{\Np}(\hmu_P - \mu) /\hsigma_P
\]
follows the Student-$t$ distribution with $\Np - 1$ degrees of freedom. Therefore, assuming the null hypothesis $H_0: \mu = 0$, the $p$-value is defined as 
\[
\bT_{\Np-1}\left( \sqrt{\Np}(\hmu_P - \mu) / \hsigma_P \right)
\]
where $\bT_{\nu}$ is the survival function of the Student-$t$ distribution with its degree of freedom $\nu$. 

Note that the conservative $p$-value (Section \ref{sec_conserv}) for the case of unknown variance is highly nontrivial: This is because selecting the top-$\Np$ does not always maximizes the $p$-value when the variance needs to be estimated. A smaller empirical variance of $\Dp$ sometimes yields a larger improvement normalized by the variance. 

Unlike the conservative $p$-value that is very nontrivial to calculate, we can extend the testing with inspector model as follows.
On the inspector model, the inspector uses the following two-sample $t$ statistics to compare the means of $\Dp$ and $\Di$.
Let $\hsigma_I$ is the corresponding unbiased variance of $\Di$. Let
\[
Z_T = \frac{\hmu_P - \hmu_I}{\hsigma_{\mathrm{weight}}^2 \sqrt{1/\Np + 1/\Ni}}
\]
where 
\[
\hsigma_{\mathrm{weight}}^2 = \frac{(\Np-1) \hsigma_P^2 + (\Ni-1) \hsigma_I^2}{\Np + \Ni - 2}.
\]
 Under the assumption that $\Dp$ and $\Di$ share a common mean, $z_T$ follows the Student-$t$ distribution of the degrees of freedom $\Np + \Ni - 2$, and thus the inspector claims the selection bias when
\[
\bT_{\Np+\Ni-2}(Z_T) \le \beta.
\]

In summary, while the test statistics in the case of known variance follow normal distributions, the test statistics in the case of unknown variance follow Student-$t$ distributions.
Although the Student-$t$ distribution is heavy-tail (i.e., the volume of the tail is polynomial to the distance from the origin) unlike the normal distribution, as $\Np, \Ni \rightarrow \infty$, it converges to the normal distribution. In this sense, it is reasonable to presume that similar results to Sections \ref{sec_cherrypick}-\ref{sec_inspector} hold in the case of unknown variances. 
In Section \ref{sec_experiment}, we conduct an empirical comparison between the cases of known and unknown variances.


\subsection{Nonparametric testing}
\label{sec_nonpara}

So far, we have assumed the normality of each $\val_i$. 
The assumption of normality is ubiquitous to the fields of the science \cite{ghasemi2012normality,normalpubhealth}, and there are many reasons to assume normality even when we are not convinced of it. Some of the most important reasons are that (i) the central limit theorem implies that the sum of independent samples converges to a normal distribution under very mild conditions.
(ii) Moreover, the normal distribution has some robustness against estimation error. In particular, given random and independent observations, the sample mean and sample variance are independent. 

When we cannot assume particular classes of distributions, there are some alternative nonparametric tests.
The nonparametric test that corresponding to our reporter (Section \ref{sec_cherrypick}) is the Wilcoxon signed rank test. 
Considering the fact that the corresponding nonparametric framework is highly nontrivial, we consider it as future work.

\begin{figure*}[t!]
\vspace{-1.5em}
\begin{center}
  \setlength{\subfigwidth}{.49\linewidth}
  \addtolength{\subfigwidth}{-.49\subfigcolsep}

  \begin{minipage}[t]{\subfigwidth}
  \centering
  \subfigure[Ratio of false claims (unknown variance).]{
    \includegraphics[scale=0.5]{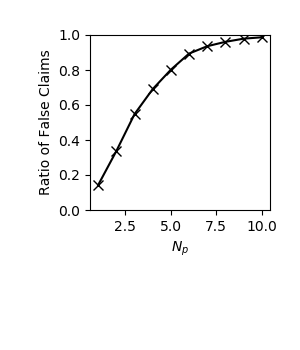}
    \label{fig_falsereport}
  }
  \end{minipage}\hfill
  \begin{minipage}[t]{\subfigwidth}
  \centering
  \subfigure[Statistical power of the inspector model (known variance).]{
    \includegraphics[scale=0.5]{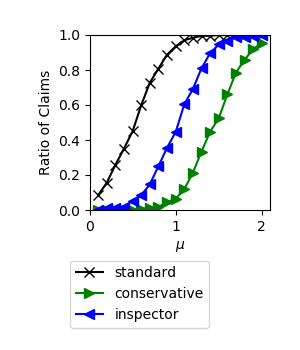}
    \label{fig_statpower}
  }
  \end{minipage}\hfill
  \begin{minipage}[t]{\subfigwidth}
  \centering
  \subfigure[Ratio of false claims (unknown variance).]{
    \includegraphics[scale=0.5]{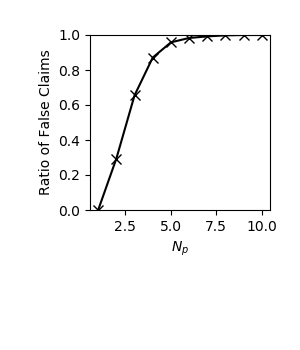}
    \label{fig_falsereport_t}
  }
  \end{minipage}\hfill
  \begin{minipage}[t]{\subfigwidth}
  \centering
  \subfigure[Statistical power of the inspector model (unknown variance).]{
    \includegraphics[scale=0.5]{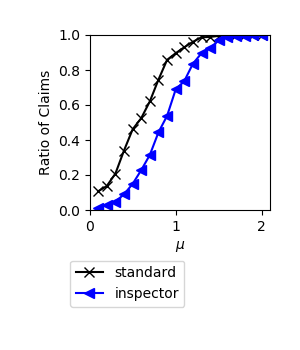}
    \label{fig_statpower_t}
  }
  \end{minipage}\hfill

\end{center}
\vspace{-1.5em}
\caption{
We consider two scenario for the reporting process: Namely, in Scenario 1 the reporter made unbiased selection, and in Scenario 2 the reporter selected the top-$\Np$ among $\{\val_i\}_{i=1}^{\Na}$. Each empirical probability was estimated with $1,000$ independent trials. Note that results in the figure are not very sensitive to the choice of $\Np, \Na$ as long as $\Np$ is smaller than $(1/2)\Na$ with some margin.
\textbf{(a):} We set $\mu=0$ and $\Np/Na = 1/3$. The figure shows the empirical probability that $p_P(\hmu) \le \alpha$ as a function of $\Np$ under Scenario 2. This value measures the ratio in which the reporter can make the false claim. 
\textbf{(b):} We set $\Np = 10$ and $\Na = 30$. Under Scenario 1, we saw the statistical power of each method. The figure shows the empirical probability that $p_P(\hmu_P) \le \alpha$ (standard), $p_P^{\conserv}(\hmu_P) \le \alpha$ (conservative), and $p_P^{\mugap}(\hmu_P) \le \alpha$ (inspector).
We set $\mugap = 0.5$, which was a sufficient value such that the inspector was able to detect the biased selection (i.e., $p_Z \le \beta$) with empirical probability larger than $0.9$ in the case of Scenario 2. 
This simulation measured the statistical power of the standard $p$-value, the conservative $p$-value, and the inspector model. 
\textbf{(c) and (d):} the corresponding experiments to (a) and (b) with unknown variance. Note that we did not derive the conservative $p$-value in the case of unknown variance because it is highly nontrivial (see Section \ref{sec_unkvariance}).
}
\vspace{-1em}
\end{figure*}%

\section{Experiment}
\label{sec_experiment}

To verify the practical performance of the proposed inspector model, we conducted simulations with synthetic datasets (Section \ref{sec_synth}) and real-world datasets (Section \ref{sec_classify}).
The goal of these simulations is to provide a quantitative view on 
(i) the vulnerability of the standard hypothesis testing procedure to the biased selection, and 
(ii) how much the conservative $p$-value and the inspector model lose the statistical power in return for the robustness to such biased selection.
The simulations are implemented on the top of the scikit-learn machine learning library\footnote{http://scikit-learn.org/}, and the source codes are going to be released in the camera-ready version.
Throughout the simulations, We set $\alpha = \beta = 0.05$.


\subsection{Synthetic data}
\label{sec_synth}

We first conducted simulations with synthetic data where $\val_i \sim \Normal(\mu, 1)$ for each $i \in \Da$. 
Firstly, we fix $\mu=0$ (i.e. no innovation) and saw the possibility of the false claim. 
Assuming that the reporter made biased selection, Figure \ref{fig_falsereport} shows the ratio that the reported $\hmu_P$ was significant improvement at level $\alpha$ as a function of $\Np$. From the figure, we can see that the larger $\Np$ (and $\Na$) is, the higher the risk of the false report is, which is consistent with Theorem \ref{thm_cherrypick}.

Secondly, we measured the statistical power sacrificed by adapting the conservative $p$-value and the inspector model. Assuming that the reporter made unbiased selection, Figure \ref{fig_statpower} shows the probability of finding the true discovery for the case $\val_i \sim \Normal(\mu, 1)$ as a function of $\mu > 0$. Although the inspector model requires positive $\mugap > 0$ (and thus sacrifices some amount of the statistical power) so that it can detect biased selection, this result suggests that the inspector model is more powerful than the conservative $p$-value in terms of the statistical power to detect true findings.

We next simulated statistical testing with unknown variance that are described in Section \ref{sec_unkvariance}. Figure \ref{fig_falsereport_t} and \ref{fig_statpower_t} show the result with unknown variance.
At a word, estimating variance sacrifices some statistical power.
In comparison with the case of known variance (Figure \ref{fig_falsereport} and \ref{fig_statpower}), one can see that 
(i) the false claim was easier for the reporter: This is mainly because that the estimated variance is smaller than the true variance when the selection is biased. 
Moreover,
(ii) detecting biased selection was harder for the inspector: This is mainly because the Student-$t$ distribution has a heavier tail than the exponentially-decaying normal distribution.

\subsection{Real data}
\label{sec_classify}

\begin{figure}[t!]
 \vspace{-1em}
 \begin{center}
 \includegraphics[scale=0.8]{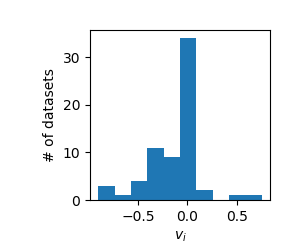}
 \end{center}
 \caption{The empirical distribution of $\val_i = \kappa_i^{\mathrm{LR}} - \kappa_i^{\mathrm{GBT}}$ in the classification datasets. LR is outperformed by GBT in the sense that the average this value for all datasets $(1/\Na)\sum_{i \in \Da} \val_i  \approx -0.13 < 0$.}
 \label{fig_cohenkappa}
\end{figure}%
To demonstrate an example of potential false claims, we studied a standard binary classification task. 
We retrieved $66$ binary classification datasets from the libSVM repository\footnote{https://www.csie.ntu.edu.tw/~cjlin/\\\hspace{5em} libsvmtools/datasets/binary.html} and the Keel dataset repository\footnote{https://sci2s.ugr.es/keel/category.php?cat=clas}. We compared the performances of two well-known algorithms: Namely, the logistic regression (LR) and the gradient boosting tree (GBT). Although GBT outperforms LR on average, in the following we show the reporter can lead to the opposite conclusion.

For a performance measure of an algorithm, we used the Cohen's $\kappa$ statistics \cite{cohenkappa} of an algorithm $\mathcal{A}$, which was defined as 
$\kappa_{\mathcal{A}} =  (s_\mathcal{A} - s_{\mathrm{base}} )/ (1 - s_{\mathrm{base}})$,
where $s_{\mathrm{base}}$ indicated the score of the baseline algorithm that classified all datapoint as the major category. The statistics $\kappa_{\mathcal{A}}$ measured how well an algorithm $\mathcal{A}$ performed compared with a native baseline.
For each dataset $i$, we defined $\val_i = \kappa_i^{\mathrm{LR}} - \kappa_i^{\mathrm{GBT}}$, which was the difference of the Cohen's $\kappa$ between the LR and GBT algorithms. Figure \ref{fig_cohenkappa} shows the empirical distribution of $\val_i$. Even though Figure \ref{fig_cohenkappa} indicated that GBT outperformed LR on average, the reporter here tried to promote the hypothesis that LR outperformed GBT, which was false in the sense that $(1/\Na)\sum_{i=1}^{\Na} \val_i < 0$.

We set $\Np = \Ni = 5$ and normalized $\val_i$ so that it had a unit variance. The reporter submitted $\Dp$ on the basis of the top-$\Np$ among $\{\val_i\}$. We obtained the following two results\footnote{Note that our results are not sensitive to the choice of $\Np$: The same discussion applies to any $5 \le \Np \le 10$.}.
First, we confirmed $p_P(\hmu_P) \approx 2.6 \times 10^{-3}$. Namely, the reporter was able to make the false claim of LR outperforming GBT.
Second, we confirmed that $p_Z \le \beta$ in $959$ out of $1,000$ trials when $\Di$ was uniformly resampled from $\Da$. In other words, with high probability, the inspector detected the biased selection by comparing $\Dp$ with $\Di$.

\section{Conclusion}
\label{sec_conclusion}

We considered the issue of selection bias in data-oriented scientific findings. When the number of datasets is large, a reporter is able to claim an improvement even though there is no improvement on average. By requiring the reporter a moderate number of reported datasets $\Np$ and a reasonable amount of improvement $\mugap$, such a false claim can be detected by comparing the published results with the results of resampled datasets.

The results of our study suggest that the effectiveness of a post-publication model in which other people verify the published results rather than an attempt to control publication bias at submission time. In verifying published results, there are several directions in which efforts may pay off, for instance, lowering the cost of replication and incentivizing reports of replicated results. In particular, the latter effort tends to be underestimated because replicating existing results is not generally considered to be as worthwhile as reporting novel research.

Possible lines of future work include: 
\begin{itemize}
\item \textbf{Nonparametric extension:} where we do not assume the normality of $\val_i$, which enables us to deal with heavy-tail distributions.
\item \textbf{Sequential tests} such as Bayeisan optimization \cite{pracbo} and A/B tests \cite{DBLP:conf/kdd/JohariKPW17} would enables us to deal us a way to collect data efficiently with statistical guarantees. Such sequential tests require more sophisticated analysis of the confidence region. 
\end{itemize}

\newpage

\bibliographystyle{alpha}
\bibliography{mybibs.bib}

\clearpage
\appendix
\onecolumn 

\section{Lemmas}

The following facts and lemmas are used in the paper.

\begin{fact}{\rm (Order statistics)}
For $i = 1,2,\dots,N$, let $X_i \sim \Unif(0,1)$ an i.i.d. random variable uniformly distributed on $[0,1]$. 
Let $X_{(i)}$ be the $i$-th largest among $\{X_i\}$. 
Then,  
\[
 X_{(i)} \sim \Beta(i, n+1-i).
\]
\label{fact_orderstats}
\end{fact}

\begin{fact}{\rm (Stirling approximation)}
The following inequality holds for any $n \in \Natural$:
\[
\sqrt{2 \pi} n^{n+1/2} e^{-n} \le n! \le e n^{n+1/2} e^{-n}.
\]
\label{fact_stirling}
\end{fact}

\begin{fact}{\rm (Binomial approximation)}
For any two natural numbers $k,n$ such that $0 \le k \le n$:
\[
\left(\frac{n}{k}\right)^k \le \binom{n}{k} \le \left(\frac{n e}{k}\right)^k
\]
holds.
\label{fact_binomial}
\end{fact}

\begin{lemma}{\rm (Mill's ratio \cite{Komatu55,Yang2015})}
For any $x>0$, 
\begin{equation}
\frac{2}{\sqrt{x^2+4}+x} \le \frac{\bPhi(x)}{\phi(x)} \le \frac{2}{\sqrt{x^2+2}+x}
\label{ineq_mills}
\end{equation}
holds.
\end{lemma}

\begin{lemma}{\rm (Sub-gaussian concentration, page 25 in \cite{boucheron2013concentration})}
A random variable $X$ is a mean-zero sub-Gaussian with variance factor $\nu>0$ if
\[
\log \Expect[\e^{\lambda X}] \le \frac{\lambda^2 \nu}{2}.
\]
For a mean-zero sub-Gaussian random variable $X$ with its variance factor $\nu>0$, the following inequalities hold for any $t > 0$:
\begin{align}
\Prob[X > t] &\le \e^{-t^2/(2\nu)} \nn
\Prob[X < -t] &\le \e^{-t^2/(2\nu)} .
\end{align}
\label{lem_subgconcentration}
\end{lemma}

\begin{lemma}{\rm (Concentration of Beta distribution)}
Let $t>0$. Theorem 1 in \cite{marchal2017} implies that the beta distribution $\Beta(\alpha, \beta)$ is sub-Gaussian with mean $\alpha/(\alpha+\beta)$ and variance factor $1/(4(\alpha+\beta+1))$. This fact, combined with Lemma \ref{lem_subgconcentration} states that $X \sim \Beta(\alpha, \beta)$ satisfies
\begin{align}
 \Pr\left[X - \frac{\alpha}{\alpha+\beta} > t\right] &\le \e^{-2(\alpha+\beta+1) t^2} \nn
 \Pr\left[X - \frac{\alpha}{\alpha+\beta} < t\right] &\le \e^{-2(\alpha+\beta+1) t^2}.
\label{ineq_betaconcentration}
\end{align}
\label{lem_betaconcentration}
\end{lemma}

The following Lemma \ref{lem_alphatail} characterizes $x$ such that $\bPhi(x) = \alpha$.
\begin{lemma}
\label{lem_alphatail}
For $\alpha \in (0,1/2)$, let $x \in \Real$ be such that $\bPhi(x) = \alpha$.
Then, $x = \Theta( (\log(1/\alpha))^{1/2} )$. 
More explicitly, we have
\begin{equation}
\log(1/\alpha) + \log(1/\sqrt{2 \pi}) - 1 < x^2 < 2 \log(1/\alpha) + \log(1/\pi).
\end{equation}
\end{lemma}
\begin{proof}[Proof of Lemma \ref{lem_alphatail}]
Ineq. \eqref{ineq_mills} implies 
\begin{equation*}
\alpha \frac{\sqrt{x^2 + 2}}{2} < \phi(x) = \frac{1}{\sqrt{2\pi}} \exp(-x^2/2),
\end{equation*}
and an elementary transformation with $\frac{\sqrt{x^2 + 2}}{2} > 1/\sqrt{2}$ leads
\[
\frac{x^2}{2} < \log(1/\alpha) + \frac{1}{2} \log(1/\pi).
\]
Moreover, Ineq. \eqref{ineq_mills} implies
\begin{equation*}
\frac{1}{\sqrt{2\pi}} \exp\left( -\frac{x^2}{2} \right) < \alpha \frac{\sqrt{x^2 + 4}}{2} < \alpha (1+x) < \alpha \exp(x).
\end{equation*}
An elemental transformation leads
\[
\frac{x^2}{2} + x > \log(1/(\sqrt{2 \pi} \alpha)),
\]
which, combined with the fact that $x^2/2 > x - 1$ implies
\[
x^2 > \log(1/(\sqrt{2 \pi} \alpha)) - 1.
\]
\end{proof}

\section{Proofs}

\subsection{Proof of Theorem \ref{thm_cherrypick}}

\begin{proof}
We consider the reporter who selects the top-$\Np$ elements among $\{\val_1,\val_2,\dots,\val_{\Na}\}$, which clearly minimizes the reported $p$-value. Let $\alpha_{i} = \bPhi(\val_i)$ and $\alpha_{(i)}$ be the $i$-th smallest among $\{\alpha_i\}_{i=1}^{\Na}$.

Note that 
\begin{equation}
 p_P(\hmu_P) = \bPhi( \hmu_P \sqrt{\Np}) \le \alpha
\end{equation}
is equivalent to
\begin{equation}
 \hmu_P \ge \sqrt{1/\Np} \bPhi^{-1}(\alpha),
\end{equation}
where $\bPhi^{-1}(\alpha) = x$ such that $\bPhi(x) = \alpha$ holds\footnote{Such an $x$ uniquely exists for $\alpha \in (0,1)$.}. 
The fact that $\hmu_P$ is the mean of the top $\Np$-th among $\{\val_i\}$ implies that its sufficient condition is 
\[
 \bPhi^{-1}(\alpha_{(\Np)}) \ge \sqrt{1/\Np} \bPhi^{-1}(\alpha),
\]
which is equivalent to
\begin{equation}
\label{ineq_cherrypick_suffcond_pre}
 \alpha_{(\Np)} 
 \le \bPhi( \sqrt{1/\Np} \bPhi^{-1}(\alpha) ).
\end{equation}
Moreover, 
\begin{align}
\lefteqn{
 \bPhi( \sqrt{1/\Np} \bPhi^{-1}(\alpha) )
 } \nn &\ge 
 \bPhi\left( \sqrt{(1/\Np) (2 \log(1/\alpha) + \log(1/\pi) )} \right) \nn
 &\ \ \ \ \text{\ \ \ (by Lemma \ref{lem_alphatail})} \nn
 & \ge \frac{1}{2} - \sqrt{(1/\Np) (2 \log(1/\alpha) + \log(1/\pi) )} \nn
 &\ \ \ \ \text{\ \ \ (by $\bPhi(0)=(1/2)$ and $\phi(x)\le 1$)} \nn
 & \ge \frac{1}{2} - \sqrt{(2/\Np) \log(1/\alpha) )} \nn
 & \ge \frac{1}{2} - \frac{\epsilon}{2} \nn
 &\ \ \ \ \text{\ \ \ (by assumption)}
\end{align}
and thus the sufficient condition to \eqref{ineq_cherrypick_suffcond_pre} is 
\begin{equation}
 \alpha_{(\Np)} \le \frac{1}{2} - \frac{\epsilon}{2}.
\label{ineq_cherrypick_suffcond}
\end{equation}
In the following, we show that \eqref{ineq_cherrypick_suffcond} holds with high probability.

Remember that each $\alpha_{i}$ for $i \in \{1,2,\dots,\Na\}$ is independent and identically distributed from the uniform distribution $\Unif([0,1])$. 
The theory of order statistics states that, $\alpha_{(i)}$, the top $i$-th element among $\alpha_{1},\dots,\alpha_{\Na}$, is drawn from the Beta distribution $\Beta(i, \Na+1-i)$ with its mean $i/(\Na+1)$.
The fact that 
\[
\alpha_{(\Np)} \sim \Beta(\Np, \Na+1-\Np),
\]
combined with Beta concentration (Lemma \ref{lem_betaconcentration}) implies that with probability $1-\delta$, we have
\begin{align}
 \alpha_{(\Np)} 
 &\le \frac{\Np}{\Na+1} + \sqrt{ \frac{\log{(1/\delta)}}{2 \Na} } \nn
 &\le \frac{1}{2} - \frac{\epsilon}{2} \nn
 &\ \ \ \ \text{\ \ \ (by assumption)},
\end{align}
which combined with \eqref{ineq_cherrypick_suffcond} completes the proof.
\end{proof}

\subsection{Proof of Theorem \ref{thm_losepower}}

\begin{proof}
From $\val_i \sim \Normal(\mu, 1)$, we have
\begin{equation}
 \Prob[\hmu_P \ge \mu + \bPhiInv(\delta)] \le \delta.
\label{ineq_losepower_ub1}
\end{equation}
On the other hand, a necessary condition for $\hmu_P$ to be significant at level $\alpha$ with the convervative $p$-value is that $\hmu_P$ exceeds the top $\alpha$-quantile of the $\Np$-th order statistics $\Beta(\Np, \Na+1-\Np)$. Note that for $\Np, \Na \ge 2, \Np < 2 \Na$, the median of $\Beta(\Np, \Na+1-\Np)$ is smaller than $r$, and thus for $\alpha \le 1/2$
\begin{multline}
\Prob[\Beta(\Np, \Na+1-\Np) < r] \ge \\
\Prob[\Beta(\Np, \Na+1-\Np) < \mathrm{Median}(\Beta(\Np, \Na+1-\Np))] \\ = 1/2 > \alpha
\end{multline}
which implies that the necessary condition of $\hmu_P$ to be significant at level $\alpha$ is $\bPhi(\hmu_P) < r$, which is equivalent to $\hmu_P > \bPhiInv(r)$.
Therefore, if 
\[
\bPhi^{-1} (r) \ge \mu + \bPhiInv(\delta)
\]
then with probability $1-\delta$, $\hmu_P$ is not $\alpha$-innovative for any $\alpha < 1/2$.
\end{proof}

\subsection{Proof of Theorem \ref{thm_mininnov}}

\begin{proof}
Since the $\Dp$ of minimum $p$-value is the one that chooses the top-$\Np$ largest $\{\val_i\}_{i=1}^{\Na}$, it suffice to consider that such a $\Dp$.

The event that $\Dp$ that is innovative with significance level $\alpha$ is equivalent to
\[
 \bPhi( (\hmu_P - \mugap ) \sqrt{\Np} ) \le \alpha,
\]
which is easily transformed into the following equivalent condition:
\begin{equation}
 \hmu_P \ge \sqrt{1/\Np} \bPhi^{-1}(\alpha) + \mugap.
\label{ineq_alphasig}
\end{equation}
In the following, we upper-bound $\hmu_P$ to show that Ineq. \eqref{ineq_alphasig} does not hold with high probability.
Note that
\[
 \hmu_P = \frac{1}{\Np} \sum_{i=1}^{\Np} \bPhi^{-1}(\alpha_{(i)}).
\]
Let $q \in (0,1)$.
Note that $\alpha_{(i)} \le q$ is represented by using the binomial coefficient as:
\begin{align}
\Prob\left[ \alpha_{(i)} \le q \right] 
 &\le \binom{\Na}{i} q^i \nn
 &\le \left(\frac{\Na e}{i}\right)^i q^i \nn
 & \text{\ \ \ (by binomial approximation)}
\label{ineq_multbinom}
\end{align}
and thus 
\[
\Prob\left[ \alpha_{(i)} \le \frac{i}{e \Na} \left(\frac{\delta}{\Np}\right)^{1/i} \right] 
\le \left( \left(\frac{\delta}{\Np}\right)^{1/i} \right)^i = \frac{\delta}{\Np}.
\]
By using this, we obtain
\begin{align}
\lefteqn{
 \Prob\left[\bigcup_{i=1}^{\Np} \left( \alpha_{(i)} \le \frac{i}{e \Na} \left(\frac{\delta}{\Np}\right)^{1/i} \right) \right]
}
 \nn &\le \sum_{i=1}^{\Np} \Prob\left[ \alpha_{(i)} \le \frac{i}{e \Na} \left(\frac{\delta}{\Np}\right)^{1/i}\right] 
 \nn &= \delta.
\label{ineq_multalpha}
\end{align}
Combined with Lemma \ref{lem_alphatail}, with probability $1-\delta$ we have
\begin{align}
\lefteqn{
 \hmu_P 
} \nn
 &\le \frac{1}{\Np} \sum_{i=1}^{\Np} \sqrt{2 \log(1/\alpha_{(i)}) + \log(1/\pi)} \nn
 &\le \frac{1}{\Np} \sum_{i=1}^{\Np} \sqrt{2 \log(1/\alpha_{(i)})} \nn
 &\le \frac{1}{\Np} \sum_{i=1}^{\Np} \sqrt{2 \log\left(  \frac{e \Na}{i}\left(\frac{\Np}{\delta}\right)^{1/i} \right)} \nn
 & \text{\ \ \ (with probability $1-\delta$, by \eqref{ineq_multalpha})} \nn
 &\le \frac{1}{\Np} \left( \sqrt{2 \Np \log\left(\frac{e \Na^2}{\delta}\right)} + \int_{\sqrt{\Np}}^{\Np} \sqrt{2 \log\left(  \frac{e \Na^{1+(1/\sqrt{\Np})}}{x \delta}\right) } dx \right) \nn
 &\le \frac{1}{\Np} \Biggl( \sqrt{2 \Np \log\left(\frac{e \Na^2}{\delta}\right)} + \Np\left( \sqrt{2 \left(\frac{1}{\sqrt{\Np}} \log(\Na) \right)} + \sqrt{2 \left( \log\left(\frac{1}{\delta}\right) \right)} \right) + \int_{\sqrt{\Np}}^{\Np} \sqrt{2 \log\left(  \frac{e \Na}{x}\right) } dx \Biggr) \nn
 &\le \frac{1}{\Np} \Biggl( 3 \sqrt{2 \Np^{3/4} \log\left(\frac{e \Na^2}{\delta}\right)}  + \int_{\sqrt{\Np}}^{\Np} \sqrt{2 \log\left(  \frac{e \Na}{x}\right) } dx \Biggr) \nn
 &= \frac{1}{\Np} \Biggl( 3 \sqrt{2 \Np^{3/4} \log\left(\frac{e \Na^2}{\delta}\right)}    + \sqrt{2} \Biggl[ x \sqrt{\log\left(  \frac{e \Na}{x}\right)} - \frac{e\Na \sqrt{\pi}}{2} \erf\left(\sqrt{\log\left(  \frac{e \Na}{x}\right)}\right) \Biggr]_{\sqrt{\Np}}^{\Np} \Biggr) \nn
 & \text{\ \ \ (by $\int \sqrt{\log(c/x)} dx = x\sqrt{\log(c/x)} - (\sqrt{\pi} c/2) \erf\left(\sqrt{\log(c/x)}\right)+ \mathrm{const}$)} \nn
 &< \frac{1}{\Np} \left( 3 \sqrt{2 \Np^{3/4} \log\left(\frac{e \Na^2}{\delta}\right)} 
   + \sqrt{2} \left[ x \sqrt{\log\left(  \frac{e \Na}{x}\right)} \right]_{\sqrt{\Np}}^{\Np} + \frac{e\Na \sqrt{\pi}}{2} \frac{\Np}{e\Na} \right) \nn
 &< \frac{1}{\Np} \left( 3 \sqrt{2 \Np^{3/4} \log\left(\frac{e \Na^2}{\delta}\right)}
   + \sqrt{2} \Np \sqrt{\log\left(  \frac{e \Na}{\Np}\right)} + \frac{\pi}{2} \right) \nn
 &< 6 \sqrt{2 \frac{\log\left(e \Na\right)}{\Np^{1/4}} } + 3\sqrt{2 \log\left(\frac{1}{\delta}\right)} + 
   \sqrt{2 \log\left(  \frac{e \Na}{\Np}\right)} + \frac{\pi}{2 \Np} \nn
 &< 3 \sqrt{2 \log\left(\frac{1}{\delta}\right)} +  7 \sqrt{2 \log\left(  \frac{e \Na}{\Np}\right)} + \frac{\pi}{2 \Np} \nn
 & = O\left(\log\left(\frac{1}{\delta}\right)\right) + O\left(\log\left(\frac{\Na}{\Np}\right)\right)
\end{align}
where $\erf(x) = (2/\pi) \int_0^x e^{-t^2} dt > 0$ is the error function.
\end{proof}

\subsection{Proof of Theorem \ref{thm_inspector}}

\begin{proof}
The fact that $\Dp$ is significant level $\alpha$ implies 
\[
\hmu_P \ge \mugap +  \sqrt{1/\Np} \bPhi^{-1}(\alpha). 
\]
With probability $1-\delta$, we have
\[
 \hmu_I \le \bPhi^{-1}(\delta),
\]
and $p_Z$ is bounded as follows: 
\begin{align}
 p_Z 
 &= \bPhi\left( \frac{\hmu_P - \hmu_I }{ \sqrt{ \frac{1}{\Np} + \frac{1}{\Ni} } } \right) \nn
 &= \bPhi\left( \sqrt{\Np/2} (\hmu_P - \hmu_I) \right) \nn
 &\le \bPhi\left( \sqrt{\Np/2} (\mugap - \sqrt{1/\Np} \bPhi^{-1}(\alpha) - \hmu_I) \right) \nn
 &\le \bPhi\left( \sqrt{\Np/2} \mugap - \sqrt{1/2} (\bPhi^{-1}(\alpha) + \bPhi^{-1}(\delta)) \right).
\label{ineq_inspector_main}
\end{align}
By assumption on $\mugap$ we have
\[
\bPhi^{-1}(\beta) \le \sqrt{\Np/2} \mugap - \sqrt{1/2} \bPhi^{-1}(\alpha) - \sqrt{1/2}  \bPhi^{-1}(\delta),
\]
which, combined with \eqref{ineq_inspector_main}, yields $p_Z \le \beta$.
\end{proof}
\end{document}